\documentclass[11pt]{article}
\usepackage{geometry}
\usepackage{graphicx}
\usepackage{url}
\usepackage{amsmath}
\DeclareMathOperator{\tr}{tr}
\usepackage{booktabs}
\usepackage{tikz}
\usetikzlibrary{quantikz}
\usepackage{hyperref}
\usepackage{dsfont}
\usepackage{amssymb}
\usepackage{amsthm}

\newtheorem{proposition}{Proposition}
\newtheorem*{proposition*}{Proposition}
\newtheorem{corollary}{Corollary}
\newtheorem{definition}{Definition}
\newtheorem{example}{Example}

\title{Superdense Coding and Stabiliser Codes with Ising-coupled Entanglement}
\author{Abel Jansma$^{1, 2}$}
\date{
    \textit{\footnotesize  $^1$Max Planck Institute for Mathematics in the Sciences, Leipzig, Germany \\
     $^2$Quantum Informatics, School of Informatics, University of Edinburgh, UK
    }\\[2ex]%
    \today}

\begin{document}

\maketitle

\begin{abstract}
A new class of quantum states is introduced by demanding that the computational measurement statistics approach the Boltzmann distribution of higher-order strongly coupled Ising models. The states, referred to as $n$-coupled states, are superpositions of even or odd parity $n$-qubit states, generalize Bell states, and form an orthonormal basis for the $n$-qubit Hilbert space. For any $n$, the states are maximally connected and locally maximally entangled. It is proven that the $n$-qubit W and GHZ multipartite entanglement classes have vanishing hyperdeterminant for all $n\geq 3$ and $n\geq 4$, respectively, and that the $n$-coupled states fall in the latter. Still, multiple novel protocols for multi-party secure dense coding and stabiliser code construction are presented, which rely on the structure of $n$-coupled states as well as symmetry-breaking phase perturbations.
\end{abstract}

\section{Introduction}

\subsection{Multipartite entanglement}
Many of the potentially useful quantum protocols consume entanglement. For example, quantum error correcting codes, quantum teleportation, superdense coding, and quantum communication all use up entanglement to achieve their goals. For this reason, entanglement is often considered a resource \cite{wootters1998quantum, Bouwmeester2000}, so entangled states are an essential ingredient of quantum protocols. This has inspired a search for quantum states with different kinds of entanglement. For $2$ qubits, the Bell basis states are maximally entangled and of central importance in virtually all quantum information protocols. However, it is not straightforward to generalise maximal entanglement for $n>2$ qubit systems. In fact, the catagolue of \textit{absolutely maximally entangled} (AME) states \cite{helwig2013absolutely, AMEstates} has shown that maximally entangled $n$-qubit systems (in the sense that the reduced density operator is proportional to the identity) only exist for $n \in \{2, 3, 5, 6\}$. Still, certain $n$-qubit states are used throughout quantum information science, where often it is not the total amount of entanglement that is relevant, but rather how the entanglement is distributed throughout the system \cite{jozsa2003role,gross2009most}. Entangled $n$-qubit states are also referred to as \textit{multipartite} entangled states, and it has been shown that there are several classes of multipartite entangled states that cannot be transformed into each other using just (stochastic) local operations and classical communication (SLOCC), so they are entangled in fundamentally different ways \cite{dur2000three,verstraete2002four,li2007slocc}. Up to four qubits, the full spectrum of SLOCC-invariant entanglement classes has been determined, but not much is known beyond that. One reason is that it is currently not known how to calculate the canonical SLOCC-invariant, the hyperdeterminant, for more than four qubits. 

Still, this has not hindered the usefulness of multipartite entangled states, as such states are used in a variety of quantum information protocols. For example, $n$-qubit GHZ and graph states are multipartite entangled states used in many quantum error-correcting codes \cite{hein2004multiparty,hein2006entanglement,schlingemann2001stabilizer,bell2014experimental}. Cluster states, a special kind of graph state, are provably different from the GHZ \cite{greenberger1990bell} and W states \cite{dur2000three}, and form the basis of measurement-based quantum computation \cite{raussendorf2001one,raussendorf2003measurement}.

\subsection{Contribution: $n$-coupled states}
The main contribution of this manuscript is the introduction of a new class of entangled $n$-qubit states, referred to as \textit{$n$-coupled states}. The states are maximally connected, locally maximally entangled, and globally maximally \textit{encoupled}, that is, any partition of the system leaves it in a maximally mixed state over other $n'$-coupled states. The $n$-coupled states are SLOCC-equivalent to GHZ states, so do not form a new multipartite entanglement class themselves. However, by introducing symmetry-breaking and nonlocal phases into the states that preserve the measurement statistics, both new and well-known protocols for secure dense coding and stabiliser codes are derived.

In addition, while it is not known how to calculate the hyperdeterminant of arbitrary $n$-qubit states for $n>4$, we present a proof that the hyperdeterminant of the W and GHZ states vanishes for all $n\geq 3$ and $n\geq 4$, respectively.

\subsection{Motivation: Classical Ising couplings}
Consider modelling the measurement statistics over two qubits as a classical Ising model. Restricting to measurements in the computational basis, the joint probability of measuring the outcomes $a$ and $b$ is given by the Boltzmann distribution:
\begin{equation}
    P(A=a, B=b) = \mathcal{Z}^{-1} \exp{\left(-J_{AB} ab\right)}
\end{equation}
where $\mathcal{Z}$ is the usual normalisation. In previous work, it has been shown that individual coupling parameters $J_{ij}$ of an arbitrary Ising model can be written in terms of conditional log-odds ratios, or sums of joint surprisals \cite{beentjes2020higher,jansma2023higher}. For example, a 2-point coupling $J_{ij}$ among two spins $\{i, j\} \subseteq S$ can be written as:
\begin{align}
    J_{ij} &= -\log \frac{p_{11} p_{00}}{p_{01} p_{10}}\\
    \intertext{where $p_{ab}=P(i=a, j=b \mid S \setminus \{i, j\}=0)$. If the probabilities correspond to the measurement statistics of a system in a state $\ket{\psi}=\alpha_{00} \ket{00} + \alpha_{01} \ket{01} + \alpha_{10} \ket{10} + \alpha_{11} \ket{11}$, then the 2-point coupling $J_{AB}$ can be written directly in terms of the amplitudes:}
    J_{AB} &= -\log \frac{|\alpha_{00}|^2 |\alpha_{11}|^2}{|\alpha_{01}|^2|\alpha_{10}|^2}
\end{align}

This implies that when the quantum state $\ket{\psi}$ is entangled, the corresponding classical model contains nonzero Ising interactions, with the maximally entangled Bell states corresponding to the limit of strong coupling. This view leads one to consider quantum states with measurement statistics corresponding to other Ising-like models. Using similar shorthand notation as above, the 3-point coupling $J_{ijk}$ among three spins $\{i, j, k\} \subseteq S$ can be written as:
\begin{align}
    J_{ijk} = -\log \frac{p_{111} p_{001} p_{010} p_{100}}{p_{000}p_{110} p_{101} p_{011}}
\end{align}
In the limit of strongly negative coupling, a sufficient (but not necessary) condition for this to be nonzero is that the distribution $p$ has support only in the numerator (or only in the denominator for positive coupling). Note that this is similar to the construction of Ising-like logic gates in \cite{jansma2023higher}. In a system with 3 variables, a 3-point interaction can thus be associated to a probability distribution $p$ that only assigns nonzero probability to the states $111, 001, 010$, and $100$. A quantum state that exactly encodes this distribution is $\ket{\psi} = \frac{1}{2}(\ket{111} + \ket{001} + \ket{010} + \ket{100})$. Note that there are many states with the same measurement statistics---the one shown being maximally symmetric. Note also that while the measurements follow Ising-statistics, this does not imply that the states are eigenstates of some quantum Ising Hamiltonian---the Ising interpretation of the measurement correlations is only used to motivate the definition of $n$-coupled states, but none of the results that follow depend on this physical interpretation. Still, it will be shown that the $n$-coupled states are SLOCC-equivalent to GHZ states, which are ground states of a quantum Ising model with pairwise interactions.

Finally, it should be noted that the quantum advantage of parity-sorted states has been studied before in the context of quantum pseudo-telepathy games \cite{brassard2004recasting}. In this manuscript we further explore the structure of these states, and the new protocols that can be constructed from them. Since the $n$-coupled states are SLOCC-equivalent to GHZ states, \textit{a priori} no advantage should be expected beyond what is possible with GHZ states, but it will be shown in Section \ref{sec:applications} that these $n$-coupled states (and symmetry-breaking perturbations) suggest new protocols for secure dense coding and stabiliser codes.

\section{Definition}
The $n$-coupled states are simply superpositions of states with either an even or odd number of $\ket{1}$s. However, here we define them more precisely to show that this definition has some nice properties. 

Let $S^n_e = \left\{ s | s \in \{0, 1\}^{\otimes n} \land (\sum_{i=1}^n s_i) \bmod 2 = 0 \right\}$, that is, $S^n_e$ is the set of binary strings of length $n$ with an even number of ones. Let $S^n_o = \{0, 1\}^{\otimes n} \setminus S^n_e$, i.e. the strings of length $n$ with an odd number of ones. As a regular expression, $S^n_e$ corresponds to $\mathcal{L}_e \cap \{0, 1\}^{\otimes n}$, where the language $\mathcal{L}_e = 0*(10*10*)*$. Similarly, $S^n_o$ corresponds to $\mathcal{L}_o \cap \{0, 1\}^{\otimes n}$, where the language $\mathcal{L}_o = 0*(10*10*)*10*$.

\begin{definition}[$n$-coupled state\label{def:ncoupled_state}]
    The even $n$-coupled state on $n$ qubits is:
    \begin{equation}
        \ket{\psi_n^+} = \frac{\sqrt{2}}{2^{n/2}} \left(\sum_{s \in S^n_e} \ket{s}\right)
    \end{equation}
    The odd $n$-coupled state on $n$ qubits is:
    \begin{equation}
        \ket{\psi_n^-} = \frac{\sqrt{2}}{2^{n/2}} \left(\sum_{s \in S^n_o} \ket{s}\right)
    \end{equation}

\end{definition}

That is, $\ket{\psi_n^+}$ is an equal superposition of all states with even parity, whereas $\ket{\psi_n^-}$ is an equal superposition of all states with odd parity. For even $n$, $\ket{\psi_n^+}$ corresponds to a strongly positive Ising coupling in the measurement statistics, while for odd $n$ $\ket{\psi_n^+}$ corresponds to a strongly negative coupling. The inverse is true for $\ket{\psi_n^-}$. Note that adding a complex phase to each term in the superposition does not change the classical distribution over computational measurements. In the definition, all phases are fixed to be $+1$. It will be shown that the two states defined above, together with those with different phases, form an orthonormal basis for the $n$-qubit Hilbert space. In general, we refer to all states that encode a uniform distribution over parity sorted states as \textit{encoupled}, writing $n$-coupled only for the states in the orthonormal basis introduced in section \ref{sec:generating_states}.

A given $n$-coupled state does not correspond to the measurement statistics of a classical Ising model with \textit{only} an $n$-point interaction. For example, a 3-coupled state that assigns equal probabilities $p$ to the measurement outcomes $111, 001, 010, 100$, and a vanishingly small probability $\epsilon$ to the other outcomes, leads to a 3-point interaction 
\begin{align}
    I_{123} &= -\log \frac{p_{111} p_{001} p_{010} p_{100}}{p_{000}p_{110} p_{101} p_{011}} = -\log \frac{p^4}{\epsilon^4} \coloneqq I
    \intertext{but this also implies the presence of 1- and 2-point interactions:}
    I_{12} &= -\log \frac{p_{110} p_{000}}{p_{010} p_{100}} = -\log \frac{\epsilon^2}{p^2} = -I/2\\
    I_{1} &= -\log \frac{p_{100}}{ p_{000}} = -\log \frac{p}{\epsilon} = I/4
\end{align}

In general, an $n$-coupled state has a distribution over measurement outcomes $s$ according to an Ising model with interactions of order $n$ and lower, where the lower order interactions are exponentially suppressed:

\begin{align}
    p(\vec{s}) = \mathcal{Z}^{-1} \exp\sum_{m=0}^{n-1}{\left( (-1)^{m}\frac{I}{2^{m}} \prod_{i=1}^{n-m} s_i\right)}
\end{align}

\begin{example}[2-coupled two qubit states] \label{ex:2coupled}
    \normalfont On two qubits, the $n$-coupled states are the positive-phase Bell states. To see this, note that $S^2_e = \{00, 11\}$ and $S^2_o = \{01, 10\}$. Therefore, the positive 2-coupled state is $\ket{\psi_2^+} = \frac{1}{\sqrt{2}} \left( \ket{00} + \ket{11} \right)$, and the negative 2-coupled state is $\ket{\psi_2^-} = \frac{1}{\sqrt{2}} \left( \ket{01} + \ket{10} \right)$.
\end{example}

\section{Generating $n$-coupled states \label{sec:generating_states}}

\begin{figure}[h!]
\begin{quantikz}
    \lstick{$\ket{0}^{\otimes n}$}  &\gate{\left(\bigotimes_{i=1}^{n-1} H_i \right) \otimes \mathbb{I}} & \gate{\bigotimes_{i=1}^{n-1} \text{CNOT}_{n-i-1, n-i}} &\qw &\qw & = \ket{\psi^+_n}\\
    \lstick{$\ket{0}^{\otimes n}$}  &\gate{\left(\bigotimes_{i=1}^{n-1} H_i \right) \otimes \mathbb{I}} & \gate{\bigotimes_{i=1}^{n-1} \text{CNOT}_{n-i-1, n-i}} & \gate{X_i} & \qw & = \ket{\psi^-_n}\\
\end{quantikz}
\caption{Two circuits to prepare the $n$-coupled states from the vacuum state. The first circuit prepares $\ket{\psi^+_n}$, and the second prepares $\ket{\psi^-_n}$, which requires $2n-2$ and $2n-1$ Clifford gates, respectively. Note that the $X_i$ in the circuit for $\ket{\psi_n^-}$ can be chosen to act on any $1\leq i \leq n$. \label{fig:encoupled_circuit}}
\end{figure}
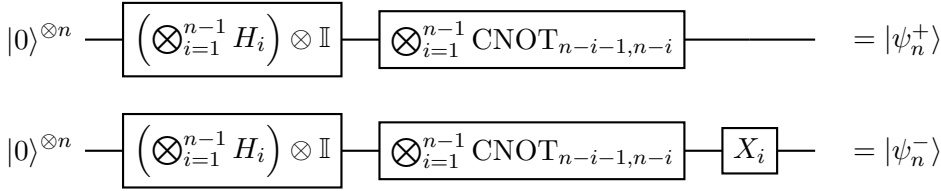

To create the $n$-coupled states from the all-zero state, one can use the circuits in Figure \ref{fig:encoupled_circuit}. Note that this is essentially a generalisation of the circuit that prepares the Bell states. In addition, one can generate a full orthonormal basis for the $n$-qubit Hilbert space by applying the circuit for $\ket{\psi^+_n}$ to each computational basis state. In that case, $\ket{\psi^-_n}$ is generated by applying the circuit to the state $\ket{1} \otimes \ket{0}^{\otimes n-1}$ instead of the vacuum state, so that the parity of the output state is exactly flipped. For $n=3$, the full orthonormal \textit{$n$-coupled basis} is then generated by the circuit $U$ as follows:

{\small{
\begin{align}
    \ket{\psi_3^1} &= U \ket{000} = \frac{1}{2}\left(~~ \ket{000}+ \ket{011}+ \ket{101}+ \ket{110}\right)\\
    \ket{\psi_3^2} &= U \ket{001} = \frac{1}{2}\left(~~ \ket{000}- \ket{011}- \ket{101}+ \ket{110}\right) = ~~~Z \otimes Z \otimes \mathds{1} \ket{\psi_3^1}\\
    \ket{\psi_3^3} &= U \ket{010} = \frac{1}{2}\left(~~  \ket{000}+ \ket{011}- \ket{101}- \ket{110}\right)= ~~~Z \otimes \mathds{1} \otimes \mathds{1} \ket{\psi_3^1} \\
    \ket{\psi_3^4} &= U \ket{011} = \frac{1}{2}\left(~~  \ket{000}- \ket{011}+ \ket{101}- \ket{110}\right)= ~~~\mathds{1} \otimes Z \otimes \mathds{1} \ket{\psi_3^1} \\
    \ket{\psi_3^5} &= U \ket{100} = \frac{1}{2}\left(~~ \ket{001}+ \ket{010}+ \ket{100}+ \ket{111}\right) = ~~~X \otimes \mathds{1} \otimes \mathds{1} \ket{\psi_3^1} \\
    \ket{\psi_3^6} &= U \ket{101} = \frac{1}{2}\left(-  \ket{001}+ \ket{010}+ \ket{100}- \ket{111}\right) = -iY \otimes Z \otimes \mathds{1} \ket{\psi_3^1} \\
    \ket{\psi_3^7} &= U \ket{110} = \frac{1}{2}\left(-  \ket{001}- \ket{010}+ \ket{100}+ \ket{111}\right) = -iY \otimes \mathds{1} \otimes \mathds{1} \ket{\psi_3^1} \\
    \ket{\psi_3^8} &= U \ket{111} = \frac{1}{2}\left(~~ \ket{001}- \ket{010}+ \ket{100}- \ket{111}\right) = ~~~X \otimes Z \otimes \mathds{1} \ket{\psi_3^1} 
\end{align}}}

It can be seen that projecting on the two subspaces spanned by $\{\ket{\psi_{1-4}}\}$ and $\{\ket{\psi_{5-8}}\}$ essentially amounts to a parity check that does not reveal any information about particular qubits, which is potentially useful in applications related to error-correction. Note also that all of these states are related through local Pauli operators on only the first two qubits. In other words: the full basis of $n$ qubits can be generated locally on $n-1$ qubits. This is obviously not true for the computational basis, but can be achieved with just $X$ and $Z$ gates for the $n$-coupled basis:

\begin{proposition}[The $n$-coupled basis can be locally generated]
    The $n$-coupled basis $\mathcal{B}_n$ can be generated from the state $\ket{\psi_n^1}$ by applying local $Z$ and $X$ Pauli operators on the first $n-1$ qubits. \label{prop:local_generation}
\end{proposition}
\noindent Proof in appendix \ref{proof:local_generation}.

Finally, it should be noted that the $n$-coupled states are SLOCC-equivalent to GHZ states, so do not form a new multipartite entanglement class themselves. In fact, they are directly related by local Clifford operations:
\begin{align}
    \ket{\psi_n^+} &= H^{\otimes n} \ket{\text{GHZ}_n}
\end{align}
where $H$ is the Hadamard gate and $\ket{\text{GHZ}_n}$ is the $n$-qubit GHZ state.

\section{Entanglement structure}
\subsection{The Schmidt decomposition of $n$-coupled states}
There is not a single measure of entanglement, and which is used depends on the situation, and what is easily calculated. The most commonly used measure is the entanglement entropy, defined as the von Neumann entropy of the reduced density operator. It turns out that any $n$-coupled state can be written as a Schmidt decomposition, with Schmidt rank 2, in lower order $n'$-coupled states. This makes calculating the entanglement entropy of any partition of the system straightforward.

\begin{proposition}\label{prop:ncoupled_schmidt}
    For any $n$ and $m$ such that $0 < m < n$, the $n$-coupled states have the following Schmidt decomposition:
    \begin{align}
        \ket{\psi_n^+} = \frac{1}{\sqrt{2}} \left( \ket{\psi_m^+} \otimes \ket{\psi_{n-m}^+} + \ket{\psi_m^-} \otimes \ket{\psi_{n-m}^-} \right)\\
        \ket{\psi_n^-} = \frac{1}{\sqrt{2}} \left( \ket{\psi_m^+} \otimes \ket{\psi_{n-m}^-} + \ket{\psi_m^-} \otimes \ket{\psi_{n-m}^+} \right)
    \end{align}
    Furthermore, this decomposition is minimal. 
\end{proposition}
\noindent Proof in appendix \ref{proof:ncoupled_schmidt}.

With this Schmidt decomposition, a number of properties are immediately clear:

\begin{corollary}
    The $n$-coupled states are locally maximally entangled \label{cor:partitioning_preserves_entanglement}
\end{corollary}
\begin{proof}
    The Schmidt decomposition of an $n$-coupled state has two positive and identical coefficients, which implies a Schmidt rank of 2, and an entanglement entropy of 1, which implies that the state is locally maximally entangled.
\end{proof}

\begin{corollary}
    The $n$-coupled states are maximally connected.\label{cor:connectedness}
\end{corollary}
\begin{proof}
In \cite{briegel2001persistent}, the authors call an entangled $n$-qubit state \textit{maximally connected} if any two qubits can be projected, with certainty, into a pure Bell state by local measurement on the other qubits. Proposition \ref{prop:ncoupled_schmidt} shows that when all but two of the qubits are measured, the remaining two qubits end up in one of the 2-coupled states. Example \ref{ex:2coupled} illustrated that these 2-coupled states are in fact Bell states. 
\end{proof}

\begin{corollary}
    The $n$-coupled states have a matrix product state representation with bond dimension 2:
    \begin{align}
        \ket{\psi_n^+} &= \frac{\sqrt{2}}{2^{\frac{n}{2}+1}} \sum_{s_1, s_2, \dots, s_n} \text{Tr}(M^{s_1} M^{s_2} \dots M^{s_n}) \ket{s_1, s_2, \dots, s_n}
    \end{align}
    where $M^0 = X$ and $M^1 = \mathds{1}_{2\times 2}$.    
\end{corollary}
\begin{proof}
    Note that $\text{Tr}\left(M^{s_1} M^{s_2} \ldots M^{s_n}\right)$ is zero if and only if an odd number of $M^0$ appear, and $2$ otherwise. This is exactly the parity of the string $s$.
\end{proof}

\begin{corollary}
    For $n>2$, the $n$-coupled states contain a Q-information of $\Omega_Q = n-3$. \label{cor:Q_info}
\end{corollary}
\begin{proof}
The authors of \cite{javaronequantifying} aim to quantify higher-order dependence in quantum states using the Q-information, defined on an $n$-qubit state $\rho$ as:
\[
    \Omega_Q = (n-2) S(\rho) + \sum_i \left( S(\text{Tr}_{n\setminus i}\rho) - S(\text{Tr}_i\rho)\right)\]
where $\text{Tr}_{n\setminus i}$ denotes the partial trace over all subsystems except for $i$. While the Q-information is easily seen to be zero for any pure state, one can associate a Q-information to a pure state by tracing out any single qubit. Since the $n$-coupled states are symmetric under qubit permutations, we trace out the first qubit without loss of generality to obtain the density matrix $\rho$, such that $S(\rho)=1$, and $\Omega_Q = (n-1)-2 = n-3$.
\end{proof}

Given the SLOCC-equivalence to GHZ states, it should not come as a surprise that these four properties are the identical to those of GHZ states. In \cite{javaronequantifying} the authors suggest that a positive Q-information reflects a redundancy in the information carried by the qubits. Qubits in a GHZ state carry redundant information in the sense that a single measurement on any qubit fixes measurement outcomes of all the other qubits to be the same. The $n$-coupled states, however, are essentially an equal superposition over the truth table of an XOR gate with $n-1$ inputs and are thus maximally synergistic with respect to measurements in the computational basis. However, measurements in the $X$-basis still show maximal redundancy. This shows that the redundancy quantified by the Q-information is basis-independent and does not correspond to a classical notion of redundancy. 

Similarly basis-independent is the persistency of the entanglement. In \cite{briegel2001persistent}, the persistency $P_e$ was defined as \textit{the minimum number of local measurements such that, for all measurement outcomes, the state is completely disentangled}. Both the GHZ states and the $n$-coupled states can be disentangled with a single measurement, so $P_e=1$. However, the entanglement of $n$-coupled states is maximally persistent when measurements are restricted to the computational basis, whereas GHZ states are disentangled after a single computational measurement.

\subsection{Hyperdeterminants}

In the same way that higher-order mutual information (also known as interaction information) has received much less attention than the bipartite case, entanglement has been mainly calculated with respect to bipartitions of the system. This is in part due to the lack of a unique measure of multipartite entanglement \cite{coffman2000distributed,verstraete2003normal,hein2004multiparty,horodecki2009quantum,miyake2004multipartite,miyake2003classification}. However, one entanglement monotone has been consistently referred to as \textit{genuine} multipartite entanglement: the absolute value of the hyperdeterminant \cite{gelfand1994hyperdeterminants,miyake2002multipartite,miyake2003classification}. The reason that it is referred to as genuine multipartite entanglement is that the SLOCC equivalence classes are orbits of the group $\text{SL}(2, \mathbb{C})^{n}$ on the combined Hilbert space $(\mathbb{C}^2)^{\otimes n}$, and the hyperdeterminant of a tensor is invariant under the action of this orbit. In fact, for $n\leq 3$, the hyperdeterminant classification is equivalent to the SLOCC classification \cite{miyake2002multipartite,miyake2003classification}. While the absolute value of the hyperdeterminant is further known to be a true entanglement monotone \cite{miyake2003classification}, calculating the hyperdeterminant for $n>4$ has proven intractable in almost all cases \cite{horodecki2009quantum,cervera2018multipartite}. Here, we show that the hyperdeterminant of W and GHZ states is zero for all $n>2$ and $n>3$, respectively. By SLOCC-equivalence the latter also holds for the $n$-coupled states. 

The hyperdeterminant is a polynomial function of the entries of a state's tensor $T$. On an $n$-qubit state $\ket{\psi}$, this $2^n = 2\times 2 \times ...$ tensor $T$ is defined as follows:
\begin{align}
    \ket{\psi} = \sum_{i_1, i_2, \ldots, i_n=0}^1 T_{i_1, i_2, \ldots, i_n} \ket{i_1 i_2 \ldots i_n}
\end{align}
The hyperdeterminant of $T$, denoted $\text{HDET}(T)$, has a relatively straightforward form for $n=2, 3, 4$, but beyond that it is a polynomial with many millions of terms. For this reason, not much is known about the hyperdeterminants of general $n$-qubit states. However, when one defines a multilinear form $f$ as
\begin{align}
    f = \sum_{i_1, i_2, \ldots, i_n} T_{i_1, i_2, \ldots, i_n} x^{(1)}_{i_1} x^{(2)}_{i_2} \ldots x^{(n)}_{i_n}
\end{align}
then the hyperdeterminant vanishes whenever both $f$ and all its partial derivatives vanish for some nontrivial set of vectors $\{x^{(i)}\}$ (nontrivial meaning  $\nexists i: x^{(i)}=\vec{0}$) \cite{cayley1845theory, gelfand1994hyperdeterminants}:
\begin{align}
    \frac{\partial f}{\partial x^{(1)}_{i_1}} = \frac{\partial f}{\partial x^{(2)}_{i_2}} = \ldots = \frac{\partial f}{\partial x^{(n)}_{i_n}} = 0 = f
\end{align}
Note that the hyperdeterminant of a product state is necessarily zero, as a product state $\ket{s}$ corresponds to a tensor with a single nonzero amplitude $t$, i.e. a global phase. The associated multilinear form is thus simply
\begin{align}
    f = t x^{(1)}_{s_1} x^{(2)}_{s_2} \ldots x^{(n)}_{s_n}
    \intertext{so that }
    \frac{\partial f}{\partial x_{s_i}^{(i)}} = t \prod_{\substack{j=1 \\ j\neq i}}^n x^{(j)}_{s_j}
\end{align}
which obviously always has nontrivial solutions, since only one component of every vector $x^{(i)}$ appears, leaving the other completely free. 

This implies that whenever an $n$-qubit state has nonzero hyperdeterminant, then it cannot be created from a product state by SLOCC, and thus contains \textit{genuine} multipartite entanglement \cite{horodecki2009quantum}. Furthermore, states with different hyperdeterminants can thus be said to have different kinds of multipartite entanglement. For $n=2, 3, 4$, the full spectrum of SLOCC equivalence classes has been mapped out (\cite{dur2000three,verstraete2002four, li2007slocc}), which clearly differentiates the GHZ states from the W states. 

Below, we show that for $n>3$, the GHZ and W states become indistinguishable from product states in terms of their hyperdeterminant.

\begin{proposition}
    Let $\ket{\text{W}_n} = \frac{1}{\sqrt{n}} \sum_{i=1}^n\left( \ket{0}^{\otimes i-1} \otimes \ket{1} \otimes \ket{0}^{\otimes n-i}\right)$. Then $\text{HDET}(\ket{\text{W}_n})=0$ for all $n>2$
\end{proposition}
\begin{proof}
    The multilinear form associated to the hyperdeterminant of the W-state is as follows:
    \begin{align}
        f(x) = \frac{1}{\sqrt{n}}\sum_{i=1}^n \left(x^{(i)}_1 \prod_{\substack{j=1 \\ j\neq i}}^n x^{(j)}_0\right)
        \intertext{so that the partial derivatives take the form}
        \frac{\partial f}{\partial x^{(k)}_1} = \frac{1}{\sqrt{n}}\prod_{\substack{j=1 \\ j\neq k}}^n x^{(j)}_{0} \label{eq:partialdiff_W_1}\\
        \frac{\partial f}{\partial x^{(k)}_0} = \frac{1}{\sqrt{n}}\sum_{\substack{i=1 \\ i\neq k}}^n \left( x^{(i)}_1 \prod_{\substack{j=1 \\ j\neq k \\ j\neq i}}^n x^{(j)}_{0}\right) \label{eq:partialdiff_W_0}
    \end{align}

    For \eqref{eq:partialdiff_W_1} to vanish for all $k$, at least two distinct $x^{(j)}_0$ have to be zero. For \eqref{eq:partialdiff_W_0} to vanish, it suffices that at least one $x^{(j)}_0$ be zero for every $i$ and $k$, where $j\neq i$ and $j \neq k$. That is satisfied as long as at least three distinct $x^{(i)}_0$ are zero, so a nontrivial solution can be constructed for all $n>2$ by setting $x^{(1)}_0 = x^{(2)}_0 = x^{(3)}_0=0$ and leaving all other variables free. This obviously also sets $f(x)=0$. Since there are nontrivial solutions on which both $f$ and its partial derivatives vanish, the hyperdeterminant vanishes. 
\end{proof}

\begin{proposition}
    Let $\ket{\text{GHZ}_n} = \frac{1}{\sqrt{2}} \left( \ket{0}^{\otimes n} + \ket{1}^{\otimes n} \right)$. Then $\text{HDET}(\ket{\text{GHZ}_n})=0$ for all $n>3$
\end{proposition}
\begin{proof}
    The multilinear form associated to the hyperdeterminant of the GHZ-state is as follows:
    \begin{align}
        f(x) = \frac{1}{\sqrt{2}} \left( \prod_{i=1}^n x^{(i)}_0 + \prod_{i=1}^n x^{(i)}_1\right)
        \intertext{so that the partial derivatives take the form}
        \frac{\partial f}{\partial x^{(i)}_s} = \frac{1}{\sqrt{2}}\prod_{\substack{j=1 \\ j\neq i}}^n x^{(j)}_{s}
    \end{align}
    For all these partial derivatives to vanish, at least two of the $x^{(i)}_1$ and two of the $x^{(i)}_0$ should be zero. For $n=3$, this is not possible without having at least one $x^{(i)} = \vec{0}$, which means that there are no nontrivial solution, and the hyperdeterminant is nonzero. In contrast, for $n>3$, one can always construct a solution as $x^{(1)}_0 = x^{(2)}_0 = x^{(3)}_1 = x^{(4)}_1 = 0$, leaving all other variables free, which makes $f$ and all partial derivatives vanish, resulting in a vanishing hyperdeterminant. 
\end{proof}
To our best knowledge, the hyperdeterminant of the GHZ and W states was not known to be zero for all $n>3$ before. Note that this immediately implies that the hyperdeterminant of the $n$-coupled states is zero for all $n>3$.

\section{Applications of Encouplement \label{sec:applications}}
Since the $n$-coupled states are maximally connected, they can always be projected to a Bell state, at which point they can be used in standard quantum information protocols. However, there is more entangled structured available and the fact that the whole basis can be generated locally leads to new capabilities. However, because the $n$-coupled states are SLOCC-equivalent to GHZ states, they cannot be used to achieve anything that was not already possible with GHZ states. Still, their precise definition suggests some other uses of the entanglement structure, and they allow for precise phase-perturbations that result in both known and new quantum protocols. The phase perturbations preserve the Ising-couplings in the measurement statistics, so we will refer to all such states as \textit{encoupled}. Some of their applications are explored in this section.

\subsection{Quantum-secure and dense coding using encouplement}
Traditionally, quantum-secure dense coding has been done with shared Bell pairs. Every shared Bell pair allows one party to send two classical bits by sending one qubit. That is, by preparing $2n$ qubits in a pairwise maximally entangled state, you can send $2n$ classical bits by sending just $n$ qubits over a quantum channel. 

\paragraph{Dense coding with shared dictionary}
The $n$-coupled states allow for a different dense coding protocol. As shown in Section \ref{sec:generating_states}, each of the eight $3$-coupled states can be generated by applying two local Pauli operators on the first two qubits to the state $\ket{\psi_3^1}$. Let Alice and Bob share the state $\ket{\psi_3^1}$, where Alice has the first two qubits, and Bob the third. Alice can then encode a 3-bit message in her two qubits by applying the two appropriate local Pauli operators. After sending the two qubits to Bob, Bob can retrieve the message by measuring the three qubits in the $n$-coupled basis, provided he knows the encoding Alice used (they share a dictionary). Alice can thus send 3 classical bits by only sending 2 qubits. That is fewer than the 4 classical bits that she could send if Alice and Bob shared two Bell pairs, but note that the total number of qubits that need to be prepared and sent is lower than in the traditional protocol. In the traditional protocol, if Claire is the one who prepares the two Bell pairs, then Claire first needs to send Alice and Bob both 2 qubits, after which Alice needs to send an additional 2 qubits to get the 4 classical bits to Bob: a total of 6 qubit transfers using 4 qubits to send 4 classical bits. In contrast, when Claire prepares the state $\ket{\psi_3^1}$, she only needs to send 2 qubits to Alice and 1 to Bob, after which Alice can send 2 qubits to Bob to send the 3 classical bits: a total of 5 qubit transfers using 3 qubits to send 3 classical bits.

A natural question is: Is this also possible for $n>3$? The most obvious choice is letting Alice and Bob share a state $\ket{\psi_4^1}$, of which Alice has two qubits, and Bob the other two. However, if Alice is to perform local Pauli operations on her two qubits, she has a total of $4^2=16$ possible choices. While that sounds like it would be enough to encode 4 classical bits, it turns out that it is not. The reason is that the 16 different Pauli operations do not lead to a new orthogonal basis. For example, $\sigma^1 \otimes \sigma ^1 \otimes \mathds{1} \otimes \mathds{1} \ket{\psi_4^1} = \ket{\psi_4^1}$. That this doesn't work is a reflection of the strong symmetry present in the definition of the $n$-coupled basis states above. However, note that the phases were arbitrarily chosen, and changing them leads to states with different symmetries, but that still correspond to Ising-coupled measurement statistics. For example, consider the state

\begin{equation}
    \ket{\chi^{00}_4} = \frac{1}{2 \sqrt{2}}\left(\ket{0000} - \ket{0011} - \ket{0101} + \ket{0110} + \ket{1001} + \ket{1010} + \ket{1100} + \ket{1111}\right)
\end{equation}

Note that this state is not invariant under the action of $\sigma^1 \otimes \sigma ^1 \otimes \mathds{1}^{\otimes 2}$, but still encodes a maximally positive 4-point Ising interaction. It was already noted in \cite{yeo2006teleportation} that using this state, one can generate the 16 orthogonal basis vectors $\ket{\chi^{ij}_4} = \sigma^i \otimes \sigma^j \otimes \mathds{1}^{\otimes 2} \ket{\chi^{00}_4}$. This means that using the state $\ket{\chi^{00}_4}$, Alice can send 4 classical bits by sending 2 qubits to Bob. This works because the two minus signs in the definition of $\ket{\chi^{00}_4}$ break the symmetry of the state, so that it is no longer symmetric under two spin flips. 

Can this technique then be extended to higher $n$? The goal is to take an $n$-coupled state $\ket{\psi_n^1}$ and add phases such that the resulting states are orthogonal even after an even number of spin flips. That means that exactly one quarter ($2^n/8$) of all terms in the superposition should get a minus sign, and that none of the terms with a minus sign should be related to each other through spin flips on the first $n/2$ qubits. When that's the case, then upon an even number of local spin flips, you end up with a state that has $2^n/8$ minus signs, and none of these coincide with those before the spin flips. The inner product then involves $2^n/4$ minus signs, and thus leads to orthogonality.

Terms that obey this are of the form $\ket{a b c \neg a \neg b \neg c}$, so one can construct a 6-qubit state by subtracting the terms that should get a minus sign twice from a 6-coupled state:

\begin{align*}
    \ket{\chi^{00}_6} = \ket{\psi_6^1} - \frac{1}{8}\Big( &\ket{000 111} + \ket{011 100} + \ket{101 010} + \ket{110 001} + \\
                                & \ket{001 110} + \ket{010 101} + \ket{100 011} + \ket{111 000}\Big)
\end{align*}

It can then easily be verified (code available from \cite{nCoupCode}) that the 64 states $\ket{\chi^{ijk}_6} = \sigma^i \otimes \sigma^j \otimes \sigma^k \otimes \mathds{1}^{\otimes 3} \ket{\chi^{00}_6}$ are all mutually orthogonal, so Alice can send 6 classical bits by sending 3 qubits to Bob. One security advantage of this protocol compared to traditional Bell-pair based dense coding, is that Bob can only read the message if he has access to all 6 qubits, whereas in the traditional protocol, Bob can read two classical bits from each Bell pair separately. This new protocol is thus more robust to eavesdroppers who intercept the qubits. 

The state $\ket{\psi_8}$ contains 128 terms, so to perform the same trick a set of 32 states has to get a phase of $-1$. These again have to be of the form $\ket{a b c d \neg a \neg b \neg c \neg d}$. However, the four free variables $a, b, c, d$ lead to $2^4=16$ possible different states that can get a minus sign, which is too few. Therefore, this trick does not work for $n=8$. In general, to create a state for $n$-qubit dense coding with the above protocol, one needs an $n$-coupled state where $2^n/8$ of the states receive a negative phase, and there are $2^{n/2}$ candidate states to receive this phase. This is indeed only possible for $n \leq 6$.

\paragraph{Dense coding without shared dictionary}

The above protocol depends on Alice and Bob sharing a dictionary that links measurement outcomes in the $n$-coupled basis to bitstrings. In this paragraph, we outline how by using $n$ qubits in an $n$-coupled state, you can send $n$ classical bits by transferring $n-1$ qubits without a shared dictionary. Note that this is a much less dense coding. In fact, the channel capacity approaches that of a classical channel as $n$ increases. Furthermore, the communication channel is only quasi-secure, in the sense that an eavesdropper can learn some bits of information from the signal. However, it has the advantage that the receiver only needs to protect a single qubit, regardless of how much data is sent over the channel. In contrast, the traditional dense coding protocol requires the receiver to protect a number of qubits that scales linearly with the number of bits to be transferred. The alternative protocol presented here might be beneficial in centralised quantum communication designs, where a central sending entity carries the responsibility for the majority of the qubits, and the receivers only need to carry around a single one. This potential for increased centralisation comes at the cost of robustness, however, since the channel is destroyed when the receiver loses their qubit.

Let Alice and Bob share the $n$-coupled state $U\ket{0}^{\otimes n} = \ket{\psi_n^+}_{A_1\ldots A_{n-1} B}$, of which Bob has access to the $n$th qubit, and Alice to all other $n-1$ qubits. In the corresponding Schmidt decomposition, this state takes the form:
\begin{align}
    \ket{\psi_n^+}_{A_1\ldots A_{n-1} B} = \ket{\psi_{n-1}^+}_{A_1\ldots A_{n-1}}\ket{1}_B +  \ket{\psi_{n-1}^-}_{A_1\ldots A_{n-1}}\ket{0}_B\\
\end{align}

Now, Alice has to find an encoding $\mathcal{E}: \{0,1\}^n \to SU(2^{n-1})$ that sends each $n$-bit message to a particular local operation on her $n-1$ qubits, such that after sending her $n-1$ qubits to Bob, Bob can retrieve the message upon applying $U^\dagger$. Proposition \ref{prop:local_generation} implies that such an encoding always exists.

The $n=2$ case reduces to the traditional dense coding protocol, so we outline here the procedure for $n=3$. Let Alice and Bob share the state $\ket{\phi}$, where Alice has the first two qubits, and Bob the third, so that
\begin{align}
    \ket{\phi_{A_1 A_2 B}} = \frac{1}{2}\left( \ket{111} + \ket{001} + \ket{010} + \ket{100} \right)
\end{align}

Alice then needs to find an encoding $\mathcal{E}: \{0,1\}^3 \to SU(2^2)$ that assigns a unique unitary operation on two qubits to each message of 3 classical bits. Let the encoding be defined according to the recipe implied by Proposition \ref{prop:local_generation}:

\begin{align}
    \mathcal{E} (000) &= \mathds{1} \otimes \mathds{1}\\
    \mathcal{E} (001) &= Z \otimes \mathds{1}\\
    \mathcal{E} (010) &= Z \otimes Z\\
    \mathcal{E} (011) &= \mathds{1} \otimes Z\\
    \mathcal{E} (100) &= X \otimes \mathds{1}\\
    \mathcal{E} (101) &= (XZ) \otimes \mathds{1}\\
    \mathcal{E} (110) &= (XZ) \otimes Z\\
    \mathcal{E} (111) &= X \otimes Z
\end{align}

Alice then encodes her message and creates the state $(\mathcal{E} \otimes \mathds{1}) \ket{\phi}$. She then sends her two qubits to Bob, who applies $U^\dagger$. Bob then measures all three qubits in the computational basis, and assuming perfect encoding and transmission, measures Alice's message with probability one.

Similarly, for $n=4$, one can define an encoding $\mathcal{E}: \{0,1\}^4 \to SU(2^3)$ as follows:

\begin{align}
    \mathcal{E} (0000) &= \mathds{1} \otimes \mathds{1} \otimes \mathds{1}  &&     \mathcal{E} (1000) = X \otimes \mathds{1} \otimes \mathds{1}\\
    \mathcal{E} (0001) &= Z \otimes \mathds{1} \otimes \mathds{1}           &&     \mathcal{E} (1001) = (XZ) \otimes \mathds{1} \otimes \mathds{1}\\
    \mathcal{E} (0010) &= Z \otimes Z  \otimes \mathds{1}                   &&     \mathcal{E} (1010) = (XZ) \otimes Z \otimes \mathds{1}\\
    \mathcal{E} (0011) &= \mathds{1} \otimes Z  \otimes \mathds{1}          &&     \mathcal{E} (1011) = X \otimes Z \otimes \mathds{1}\\
    \mathcal{E} (0100) &= Z \otimes Z \otimes Z                             &&     \mathcal{E} (1100) = (XZ) \otimes Z \otimes Z\\
    \mathcal{E} (0101) &= \mathds{1} \otimes Z  \otimes Z                   &&     \mathcal{E} (1101) = X \otimes Z \otimes Z\\
    \mathcal{E} (0110) &= \mathds{1} \otimes \mathds{1} \otimes Z           &&     \mathcal{E} (1110) = X  \otimes \mathds{1} \otimes Z\\
    \mathcal{E} (0111) &= Z \otimes \mathds{1} \otimes Z                    &&     \mathcal{E} (1111) = (XZ) \otimes \mathds{1} \otimes Z
\end{align}
Upon applying $U^\dagger$, Bob can retrieve Alice's message with probability one. 

Now, as already mentioned, as $n$ grows, the capacity of this channel approaches that of a classical channel (though it is always strictly larger). However, for any $n$, it is still quasi-secure, in the sense that anyone who intercepts the qubits that Alice sends to Bob, is left with the mixed state:
\begin{align}
    \rho_{A_1 A_2} &= \tr_B \rho_{A_1 A_2 B} = \frac{1}{2} \left( \ket{\psi_{n-1}^i} \bra{\psi_{n-1}^i} + \ket{\psi_{n-1}^j} \bra{\psi_{n-1}^j} \right)    
\end{align}
where $i$ and $j$ depend on the message that Alice sent and the encoding used. This is a mixed state, so it is not possible to retrieve the joint state with certainty. However, it is not maximally mixed over the whole computational basis, which means that by performing a measurement in the $n$-coupled basis, an eavesdropper Eve can measure the state to be in $\ket{\psi_{n-1}^i}$ or in $\ket{\psi_{n-1}^j}$ with probability $\frac{1}{2}$ each. However, actually retrieving the message from the intercepted qubits is only possible upon applying $U^\dagger$ to the combined system, which Eve does not have access to.

This protocol thus allows for quasi-secure quantum communication and might offer an advantage over the Bell-pair based approach in situations where there is an asymmetry in the ability to store qubits among the two parties. Note also that the protocol reduces to the Bell-pair method for $n=2$, so this approach can be seen as a different generalisation of the traditional dense coding protocol.

\subsection{Encoupled states are stabiliser codes}

Entanglement is also commonly used in quantum error correction. The central idea is to spread the information in a single logical qubit over a collection of multiple physical qubits, in such a way that certain errors can be detected, and possibly corrected. Note that the encoupled states represent an obvious way to `spread out' information over a number of qubits, since strong synergy corresponds to a delocalisation of the information. One common way to create and categorise quantum error correcting codes is through the stabiliser group that generates them. 

Let $P_n$ be the Pauli group over $n$ qubits, that is, $P_n = \{\pm 1, \pm i\} \times \{\mathds{1}, X, Y, Z\}^{\otimes n}$. Let $\ket{\psi}$ be any quantum state of at least $n$ qubits. Then an operator $S \in P_n$ is said to \textit{stabilise} $\ket{\psi}$ iff $S \ket{\psi} = \ket{\psi}$. The \textit{stabiliser group} of $\ket{\psi}$ is then defined as the set of all operators that stabilise $\ket{\psi}$. Note that this group is necessarily Abelian. Given a code $\mathcal{C}$ such that $C(\ket{0}) = \ket{\overline{0}}$ and $C(\ket{1}) = \ket{\overline{1}}$, the stabiliser group of $\mathcal{C}$ is the set of all operators that stabilise all vectors in the space spanned by $\ket{\overline{0}}$ and $\ket{\overline{1}}$. This group is interesting, because given an error $E \in P_n$, the error $E$ can be detected if and only if $E$ does not commute with all elements of the stabiliser group. 

Now consider using the encoupled states as codewords. For example, set $\ket{\overline{0}} = \ket{\psi_n^+}$ and $\ket{\overline{1}} = \ket{\psi_n^-}$ for $n>2$. Now the stabiliser group for this code is easily seen to be generated by the set $S$ of all operators that involve an even number of $X$-operators. No $Z$ or $Y$ operator can appear, since those would introduce phases into just some terms in the superposition, and therefore don't stabilise the states. Define the \textit{syndrome} of an error $E_k$ relative to a code $\mathcal{C}$ as a vector $\gamma^k$ of length $\dim(\mathcal{C})$ such that $\gamma^k_i = 1$ iff $[E_k, S_i]\neq 0$ and $\gamma^k_i = 0$ otherwise, where $S_i \in \mathcal{S}$, the set of generators of the stabiliser group. A code can detect an error $E_k$ for which $\gamma^k \neq \vec{0}$, and correct errors that have a unique syndrome.

Clearly, all $X$ operators commute with $\mathcal{S}$, so $X$-errors are not detectable. In contrast, errors that comprise a single $Z_i$ operator are detectable and correctable. To see this, note that for any two $Z$-errors $Z_k$ and $Z_l$, an element $S_i \in \mathcal{S}$ can be found such that $\gamma_i^k = [Z_k, S_i] \neq 0$ and $\gamma_i^l = [Z_l, S_i] = 0$, namely $S_i=X_k X_p$ where $p\neq l$ (such an element exists if and only if $n>2$). Therefore, every $Z$-error has a unique syndrome, and is thus detectable and correctable. This code makes the implementation of logical gates using physical gates very straightforward: $\overline{X}$ is implemented by applying an $X$ to any of the qubits, and $\overline{Z}$ is implemented by a parity check measurement on all qubits. 

To create a more powerful code, we need to expand the stabiliser group by choosing a different codespace. If we keep $\ket{\overline{0}}=\ket{\psi_n^+}$ but now choose $\ket{\overline{1}}=Z_n \ket{\psi_n^+}$, then none of the operators in the stabiliser group contain an $X_n$, so $[Z_n, S_i]=0$ for all $S_i \in S$, which means that $Z_n$ errors are undetectable, which is obviously not an improvement. Adding another phase $\ket{\overline{1}}=Z_{n-1} \otimes Z_n \ket{\psi_n^+}$ results in stabilisers that contain either both $X_{n-1}$ and $X_n$, or neither, so that $[Z_{n-1}, S_i] = [Z_n, S_i]$ for all $S_i \in S$, which means that the $Z_{n-1}$ and $Z_n$ errors are detectable, but not correctable. Adding a third phase $\ket{\overline{1}}=Z_{n-2} \otimes Z_{n-1} \otimes Z_n \ket{\psi_n^+}$ results in stabilisers that contain $X_{n-2} X_{n-1}$ or $X_{n-1} X_{n}$ or $X_{n-2} X_{n}$, or neither of these three. Therefore, $Z_{n-2}$, $Z_{n-1}$, and $Z_{n}$ errors are correctable. However, note that the requirement of an even total number of $X$ operators means that for $n < 6$, either none of the $X_{<n-2}$ operators appear, or all of them do, leading to $Z_{<n-2}$ errors with identical syndromes. Therefore, only for $n\geq 6$ are all $Z$-errors correctable under this code. This is not yet an improvement over the previous code, but note that the operator $\pm \bigotimes_{i=1}^n Z_i$ is now also a generator of the stabiliser group (with a $+$ for even $n$, and a $-$ for odd $n$). This means that there is now a unique syndrome, namely the one with only a 1 corresponding to the operator $\bigotimes_{i=1}^n Z_i$, that reveals that \textit{some} $X$-error has occurred. Since every $X$ error has the same syndrome, it is not possible to correct the error, but it is possible to detect it, while also being able to correct every $Z$-error, which is an improvement over the previous code that was only able to correct $Z$-errors. The implementation of logical gates is less straightforward, however, since $\overline{X}$ can now implemented as $Z_{n-2} \otimes Z_{n-1} \otimes Z_n $, but it is not immediately clear how to construct $\overline{Z}$.

Note that the smallest code that can correct both kinds of errors---the 5-qubit code---is based on the following encoding:
\begin{align}
    |0_{\rm {L}}\rangle ={\frac {1}{4}}[|00000\rangle +|10010\rangle +|01001\rangle +|10100\rangle +|01010\rangle -|11011\rangle -|00110\rangle -|11000\rangle \nonumber\\
     -|11101\rangle -|00011\rangle -|11110\rangle -|01111\rangle -|10001\rangle -|01100\rangle -|10111\rangle +|00101\rangle ]\\
    |1_{\rm {L}}\rangle ={\frac {1}{4}}[|11111\rangle +|01101\rangle +|10110\rangle +|01011\rangle +|10101\rangle -|00100\rangle -|11001\rangle -|00111\rangle \nonumber\\
    -|00010\rangle -|11100\rangle -|00001\rangle -|10000\rangle -|01110\rangle -|10011\rangle -|01000\rangle +|11010\rangle ]
\end{align}
Both these codewords are parity-sorted encoupled states. However, they do not appear in the $n$-coupled basis as generated above: the phases of each term are such that the symmetry is broken. This suggests further codes might be constructed from $n$-coupled states by breaking the symmetry in the phases, but this is left for future work.

\section{Discussion}

In this manuscript, we introduced a new class of $n$-qubit states, referred to as $n$-coupled states because their computational measurement statistics contain up to $n$-point Ising couplings. Their entanglement structure inspired new quantum information protocols by generalising Bell-pair based approaches to multipartite states. Two new superdense coding protocols were shown to offer extra protection against eavesdropper attacks relative to Bell-pair based approaches. Stabiliser codes based on the $n$-coupled states allowed for the detection and correction of $Z$-errors, and in some cases also $X$-errors.

Further hints that the $n$-coupled states contain useful entanglement comes from the literature, where some special cases of $n$-coupled states have appeared previously. For example, the Bell basis is simply the 2-coupled basis. The authors of \cite{jaffali2023maximally} were led by a numerical study to the 3-coupled state $\ket{\psi_3^+}$, as it was found to maximally violate Bell-like inequalities and correspond holographically to a non-BPS black hole \cite{kallosh2006strings,duff2008black}. The 4-coupled state $\ket{\psi_4^+}$ appears as the state $G_{abcd}$ from \cite{verstraete2002four} upon setting $d=c=0$, where it was found to maximise the hyperdeterminant. Beyond these, parity-sorted encoupled states that correspond to $n$-coupled states with added phases have appeared throughout the literature. For example, a 4-qubit encoupled state was used in \cite{yeo2006teleportation} in a new teleportation protocol, and the 5-qubit \textit{perfect} code is a 5-qubit encoupled state with symmetry-breaking phases. 
To further relate the $n$-coupled states to other entangled states, it would be interesting to understand their relationship to graph states. Stabiliser states and graph states are related through local Clifford operations \cite{grassl2002graphs,schlingemann2001stabilizer}, but it is not immediately clear how to relate the $n$-coupled states to graph states in this way. Another interesting direction for future work would be to define and analyse the more general class of encoupled qudit states, which could be straightforwardly defined by using the generalisation of classical Ising interactions to categorical variables \cite{beentjes2020higher,jansma2023higher}. 

It is not yet fully clear how higher-order Ising interactions in the measurement statistics are related to multipartite entanglement. A higher-order interaction describes a higher-order dependency in the measurement data that cannot be decomposed into lower-order quantities. For example, when there is a nonzero 3-point interaction $J_{ABC}$, then the joint probability $P(A, B, C) \neq P(A)P(B)P(C)$ but $P(A, B, C) \neq P(A, B)P(C)$ as well (and similarly for all permutations of the three variables). Analogously, if a quantum state is to be genuinely multipartite entangled across a system $S$, then we might demand that the state is not separable with respect to any partition of $S$, that is, there is no partition $\sigma(S)$ such that $\ket{\psi_{S}} = \bigotimes_{t \in \sigma(S)} \ket{\psi_t}$, except for the one-element partition. The same separability constraints then hold for the density operator. In both the classical and the quantum case, these factorisation conditions thus capture the intuition that the higher-order structure in a system is the extent to which the marginals of partitions of the system do not capture the joint distribution. Still, higher-order measurement statistics are not the same as entanglement, as an $n$-qubit GHZ state has vanishing $n$-point interactions for odd $n$, but is famously entangled.

Finally, the encoupled states are characterised by their measurement statistics. This suggests that these states might be useful in quantum machine learning, where the measurement statistics of quantum states are used to perform quantum-enhanced machine learning tasks. If a parametrized quantum circuit, for example, is trained on data from a classical Ising model with exponentially suppressed $n$-point interactions, convergence would amount to the circuit preparing an encoupled state.
\section*{Acknowledgements}
The author is grateful to Hans Briegel and Gemma De les Coves for insightful conversations on the topic of persistent entanglement, stabiliser states, and tensor decompositions. The author also thanks Elham Kashefi for pointing out the connection to graph and cluster states. Of further help were valuable discussions with Sam Kuypers and Albert Gasull on multipartite entanglement. The author finally thanks Bernd Sturmfels, Dmitrii Pavlov, and Maximilian Wiesmann for a number of helpful conversations on hyperdeterminants.

\appendix
\section{Proof of Proposition \ref{prop:local_generation} \label{proof:local_generation}}
\begin{proof}
    The whole basis $\mathcal{B}_n$ can be separated into two sectors: $\mathcal{B}^+$, containing superpositions of states that have even parity, and $\mathcal{B}^-$ containing those with odd parity. Consider the even-parity sector $\mathcal{B}^+$ first, and let $\ket{\psi_n^1}$ be the state where every term in the superposition has coefficient $+1$. With just local $Z_i$ operations, all of $\mathcal{B}^+$ can be generated, up to global phases. To show this, it suffices to show that one can generate $2^{n-1}$ mutually orthogonal vectors from $\ket{\psi_n^1}$. This can be done as follows. Consider two nonempty subsets of indices $S, T \subseteq \{1, \ldots, n-1\}$. Consider the expectation value of the operator $\left(\bigotimes_{i \in S}Z_i\right) \left(\bigotimes_{j \in T}Z_j\right)$ in the state $\ket{\psi_n^1}$. Writing $\ket{\psi_n^+}$ for $\ket{\psi_n^1}$ and $\ket{\psi_n^-}$ for $\ket{\psi_n^{2^n-1}}$, we have:
    {\small{
    \begin{align}
        \bra{\psi_n^1}\left(\bigotimes_{i \in S}Z_i\right) \left(\bigotimes_{j \in T}Z_j\right)\ket{\psi_n^1} &\propto \left(\bra{\psi_{n-1}^+}\bra{1} + \bra{\psi_{n-1}^-}\bra{0}\right) \left(\bigotimes_{i \in S}Z_i\right) \left(\bigotimes_{j \in T}Z_j\right) \left(\ket{\psi_{n-1}^+}\ket{1} + \ket{\psi_{n-1}^-}\ket{0}\right)\\
        &= \bra{\psi_{n-1}^+} \left(\bigotimes_{i \in S}Z_i\right) \left(\bigotimes_{j \in T}Z_j\right) \ket{\psi_{n-1}^+} + \bra{\psi_{n-1}^-} \left(\bigotimes_{i \in S}Z_i\right) \left(\bigotimes_{j \in T}Z_j\right) \ket{\psi_{n-1}^-}\\
        \intertext{Now note that $X_m \ket{\psi_n^+} = \ket{\psi_n^-}$ for any choice of $m$, so that}
        &= \bra{\psi_{n-1}^+} \left(\bigotimes_{i \in S}Z_i\right) \left(\bigotimes_{j \in T}Z_j\right)\ket{\psi_{n-1}^+} + \bra{\psi_{n-1}^+}X_m \left(\bigotimes_{i \in S}Z_i\right) \left(\bigotimes_{j \in T}Z_j\right) X_m \ket{\psi_{n-1}^+}\\
        &= \bra{\psi_{n-1}^+}\left(\bigotimes_{i \in S}Z_i\right) \left(\bigotimes_{j \in T}Z_j\right)  + X_m \left(\bigotimes_{i \in S}Z_i\right) \left(\bigotimes_{j \in T}Z_j\right) X_m \ket{\psi_{n-1}^+}\\
        \intertext{Now whenever $S \neq T$, it is possible to choose an $m$ that is either in $S$ or in $T$ but not in both, that is, choose $m \in (S\setminus (S \cap T)) \cup (T\setminus (S \cap T))$. Then, since $[X_i, Z_j]= 0$ iff $i \neq j$, one can pull the $X_m$ through all but one of the terms in the tensor products. Finally, note that $X_m Z_m X_m = - Z_m$, so that}
        &= \bra{\psi_{n-1}^+} \left(\bigotimes_{i \in S}Z_i\right) \left(\bigotimes_{j \in T}Z_j\right)   - \left(\bigotimes_{i \in S}Z_i\right) \left(\bigotimes_{j \in T}Z_j\right)  \ket{\psi_{n-1}^+} = 0
    \end{align}}}
    Therefore, each unique set of $Z_i$'s on $n-1$ qubits generates a new orthonormal vector. Since there are $2^{n-1}$ such sets, this generates $2^{n-1}$ mutually orthogonal vectors, i.e. an orthonormal basis for $\mathcal{B}^+$. Now, by applying a single $X_i$ to the resulting states in $\mathcal{B}^+$, the parity of every term in the superposition is flipped, which means that you can generate the whole basis $\mathcal{B}^-$ in the same way: orthogonality with $\mathcal{B}^+$ is by definition, and mutual orthogonality is inherited because the exact same reasoning as above holds. Therefore, the whole basis $\mathcal{B}_n$ can be generated from $\ket{\psi_n^1}$ by applying local Pauli operators on the first $n-1$ qubits.
\end{proof}

\section{Proof of Proposition \ref{prop:ncoupled_schmidt} \label{proof:ncoupled_schmidt}}
\begin{proof}
    Consider a partitioning of the $n$ qubits into two subsystems of size $m$ and $n-m$, respectively. Since the $n$-coupled states are fully symmetric under qubit permutation, one can take the first subsystem to be the first $m$ qubits, without loss of generality. Note that any term with an even number of ones can be constructed by composing two systems with an even number of ones, or two systems with an odd number of ones. The distributive property of the tensor product then allows us to write $\ket{\psi_n^+}$ as:
    \begin{align}
        \ket{\psi_n^+} = \frac{\sqrt{2}}{2^{n/2}} \left( \sum_{s \in S^{m}_e}\ket{s} \otimes \sum_{t \in S^{n-m}_e} \ket{t} + \sum_{s \in S^{m}_o}\ket{s} \otimes \sum_{t \in S^{n-m}_o} \ket{t} \right)
    \end{align}
    Now, we can directly insert Definition \ref{def:ncoupled_state}, and rewrite this as:
    \begin{align}
        \ket{\psi_n^+} = \frac{1}{\sqrt{2}} \left( \ket{\psi_m^+} \otimes\ket{\psi_{n-m}^+} + \ket{\psi_m^-} \otimes\ket{\psi_{n-m}^-}\right)
    \end{align}

    To write down a similar expression for decompositions of $\ket{\psi_n^-}$, note that states with an odd number of ones can only be constructed by composing a system with an even number of ones, and one with an odd number of ones. That means that $\ket{\psi_n^-}$ can be written as:
    \begin{align}
        \ket{\psi_n^-} = \frac{\sqrt{2}}{2^{n/2}} \left( \sum_{s \in S^{m}_e}\ket{s} \otimes \sum_{t \in S^{n-m}_o} \ket{t} + \sum_{s \in S^{m}_o}\ket{s} \otimes \sum_{t \in S^{n-m}_e} \ket{t} \right)
    \end{align}
    As before, this reduces to 
    \begin{align}
        \ket{\psi_n^-} = \frac{1}{\sqrt{2}} \left( \ket{\psi_m^+} \otimes\ket{\psi_{n-m}^-} + \ket{\psi_m^-} \otimes\ket{\psi_{n-m}^+}\right)
    \end{align}
    Since the $\ket{\psi_n^\pm}$ are part of an orthonormal basis for the Hilbert space of $n$ qubits, the above expressions are proper Schmidt decompositions. Finally, to show that it is in fact the minimal decomposition it suffices to show that $\ket{\psi_n^\pm}$ cannot be written as a product state. To see this, note that the density matrix can be written out as:
    \begin{align}
        \rho_n^+ &= \ket{\psi_n^+}\bra{\psi_n^+}\\
        &= \frac{1}{2} \Bigg( \ket{\psi_m^+} \bra{\psi_m^+} \otimes \ket{\psi_{n-m}^+} \bra{\psi_{n-m}^+} + \ket{\psi_m^-} \bra{\psi_m^-} \otimes \ket{\psi_{n-m}^-} \bra{\psi_{n-m}^-}\\
        &+ \ket{\psi_m^+} \bra{\psi_m^-} \otimes \ket{\psi_{n-m}^+} \bra{\psi_{n-m}^-} + \ket{\psi_m^-} \bra{\psi_m^+} \otimes \ket{\psi_{n-m}^-} \bra{\psi_{n-m}^+} \Bigg)
        \intertext{So that upon tracing out system $B$, we are only left with the diagonal terms:}
        \rho_{A}^+ &= \tr_B \rho_n^+ = \frac{1}{2} \left( \ket{\psi_m^+} \bra{\psi_m^+} + \ket{\psi_m^-} \bra{\psi_m^-} \right)
    \end{align}
    This operator has rank unequal to one, so $\ket{\psi_n^+}$ cannot be written as a product state. The proof is the same for $\ket{\psi_n^-}$, but with parity reversed. Therefore, the Schmidt decompositions of $\ket{\psi_n^\pm}$ are minimal. 
\end{proof}

\bibliographystyle{abbrv}
\bibliography{refs}

\end{document}